\documentclass[runningheads]{llncs}

\usepackage{amsmath,amssymb,mathtools}
\usepackage{mathrsfs}
\usepackage{listings}
\usepackage{booktabs}
\usepackage{xcolor}
\usepackage{enumitem}
\usepackage{tikz}
\usepackage{cite}
\usetikzlibrary{arrows.meta,positioning}
\usepackage[colorlinks=true, linkcolor=blue!70!black, citecolor=blue!70!black, urlcolor=blue!70!black]{hyperref}
\usepackage{verified-badges}

\lstdefinelanguage{ZipperGen}{
  keywords={def,workflow,act,msg,if,else,while,return,var,true,false,and,or,not},
  keywordstyle=\bfseries\color{blue!70!black},
  identifierstyle=\color{black},
  stringstyle=\color{red!70!black},
  commentstyle=\itshape\color{gray},
  morecomment=[l]{//},
  morestring=[b]",
  sensitive=true
}

\lstdefinelanguage{Monitor}{
  keywords={procedure,if,then,else,for,each,in,do,end,return,case,of,true,false,and,or,not},
  keywordstyle=\bfseries\color{blue!70!black},
  identifierstyle=\color{black},
  commentstyle=\itshape\color{gray},
  morecomment=[l]{//},
  sensitive=false
}
\AtBeginDocument{}
\newcounter{cpllisting}

\newcommand{\Lifelines}{\mathscr L}
\newcommand{\Eval}{\mathsf{Eval}}
\newcommand{\Events}{E}

\newcommand{\pid}{\mathsf{pid}}

\newcommand{\val}{\mathsf{val}}
\newcommand{\kind}{\mathsf{kind}}
\newcommand{\msg}{\lhd}
\newcommand{\causal}{\leq}
\newcommand{\strictcausal}{<}

\newcommand{\last}{\mathsf{last}}
\newcommand{\sub}{\mathsf{sub}}
\newcommand{\vars}{\mathsf{vars}}
\newcommand{\vc}{\mathsf{vc}}
\newcommand{\view}{\mathsf{view}}
\newcommand{\var}{\mathsf{var}}
\newcommand{\old}{\mathsf{old}}
\newcommand{\Past}{\mathsf{P}}
\newcommand{\Prev}{\mathsf{Y}}
\newcommand{\At}[2]{\mathord{\mathtt{@}}#1(#2)}
\newcommand{\AtField}[2]{\mathord{\mathtt{@}}#1.#2}
\newcommand{\GuardAt}[2]{#1\mathord{\mathtt{@}}#2}
\newcommand{\PastAt}[2]{\mathsf{P}_{#1}(#2)}
\newcommand{\Seen}[1]{\mathsf{Seen}_{#1}}
\newcommand{\Since}{\mathbin{\mathsf{S}}}

\newcommand{\lastloc}{\mathsf{last}_{\mathsf{loc}}}
\newcommand{\true}{\mathsf{true}}
\newcommand{\false}{\mathsf{false}}
\newcommand{\sem}[1]{\lbrack\!\lbrack #1 \rbrack\!\rbrack}

\begin{document}
\pagestyle{plain}

\title{Causal Past Logic for Runtime Verification of Distributed LLM Agent Workflows}
\titlerunning{Causal Past Logic for Runtime Verification}

\author{Benedikt Bollig}

\institute{Université Paris-Saclay, CNRS, ENS Paris-Saclay, LMF, Gif-sur-Yvette, France}

\maketitle

\begin{abstract}
We study runtime monitoring for distributed LLM-agent workflows. In an asynchronous execution, a decision can only depend on events that are causally visible to the lifeline that makes it: an event that appears earlier in some log may still be unknown locally. We extend the ZipperGen agent-workflow framework with Causal Past Logic (CPL), an adaptation of PT-DTL to guards in if-constructs and while loops. In addition to standard past-time modalities such as previous and since, a guard can inspect the latest causally visible event of another lifeline and selected variables stored there. The owner evaluates the guard online to select the next branch or loop step. We adapt the knowledge-vector monitor to ZipperGen and prove that the locally computed monitor value coincides with the denotational semantics of the guard at the current event.
\end{abstract}

\section{Introduction}
\label{sec:introduction}

Several agents based on large language models (LLMs) can cooperate on a task,
each with its own tools and local context.
Suppose a patch is changed after it has passed the tests. The earlier test result
should not justify merging the new version. The committer must therefore check
which version the result concerns, using the information it has received so far.

We use runtime verification \cite{LeuckerS09,BauerLS11,BartocciFFR18} to evaluate temporal conditions at workflow decisions. Classical runtime verification often uses temporal specifications and raises an alarm when an execution violates them. Temporal guards can also guide distributed workflow enactment~\cite{Singh03}. For example, a committer lacking a clean test result in its causal past can block the merge or ask a human reviewer to decide.

We develop this idea for ZipperGen~\cite{BolligFN2026}. ZipperGen programs
are written as global workflows and projected to local agent programs that
communicate asynchronously. ZipperGen follows the choreographic
approach~\cite{CarboneM13,Montesi2023,ShenKK23,HirschG22},
where the programmer writes the coordination from a global point of view, and the
framework derives the local behavior of the individual agents. We use the
term lifeline for the local participant, usually an agent, that executes its
part of the workflow. In ZipperGen, branch and loop decisions have an
explicit owner lifeline. The projection inserts control broadcasts so that
the lifelines that need to know a decision are informed about it.
It coordinates choices, but remote events may still be unknown to
the lifeline making a decision.


We extend ZipperGen guards with Causal Past Logic (CPL), a PT-DTL-style
past-time logic~\cite{SenVAR04} interpreted over message sequence charts
(MSCs). A guard is written as \(\GuardAt{\varphi}{A}\) and is evaluated locally by lifeline \(A\), which owns the corresponding branch or loop. CPL formulas can express causal-temporal facts such as ``the causal past contains a test pass that remains valid at the latest visible TestRunner event''. They can also combine such facts with identity checks, for instance whether the visible test result concerns the same merge candidate as the current proposal.

\paragraph{Contributions.}
We adapt PT-DTL~\cite{SenVAR04} to source-level ZipperGen guards. At an
owner-side choice event, a guard may inspect the latest causally visible event
of another lifeline and selected variables stored there. Its value selects the
next branch or loop step. The monitor follows the knowledge-vector scheme
of~\cite{SenVAR04}, using full vector clocks and latest-value views.
A correctness proof for formula-only \mbox{ptDTL} monitoring is given
in~\cite{ScheffelS14}. We prove that the monitor used here agrees with the
denotational semantics of CPL, implement the construction in ZipperGen
(\url{https://zippergen.io}), and measure its runtime and metadata costs.

\paragraph{Mechanization.}
We mechanized the definitions of MSCs, CPL satisfaction, the local
evaluator, latest-value monitor state, the operational event update, and
the monitor-correctness theorem in Lean~4~\cite{Moura2021}.
Lean represents monitor views as extensional functions, using propositions for
formula values and local event indices for clock components.
Finite-table refinement, executable Boolean evaluation, and the size bounds
are outside the mechanization. Badges next to definitions and results open the
corresponding source files.

\paragraph{Outline.}
Section~\ref{sec:overview} motivates the approach with an example. In Section~\ref{sec:executions}, we define executions in terms of MSCs. Section~\ref{sec:cpl} defines Causal Past Logic. Section~\ref{sec:monitoring} presents the online monitoring algorithm and proves correctness. Section~\ref{sec:evaluation} describes the implementation and gives a preliminary evaluation, and Section~\ref{sec:related} discusses related work. We conclude in Section~\ref{sec:conclusion}.

\section{A Distributed Code Review Example}
\label{sec:overview}

In an LLM agent workflow, agents exchange messages over dedicated channels
and maintain separate local memories. In the formal development, we call each
agent a \emph{lifeline}. In this section, we use the two terms
interchangeably.

Consider a coding-agent workflow in which a patch is checked by several
specialized agents. A TestRunner executes the test suite, a Security agent
checks for vulnerabilities, and a Committer decides whether a merge
candidate may be merged. These agents run concurrently and report their
results asynchronously. They may also send updated assessments after an
initial result. For example, the required pre-merge checks may pass, while
a subsequent check detects a regression. Similarly, a
security scan may first clear a patch and later report a vulnerable
dependency.

The Committer may merge a pull request only if, in its causal past, the
required tests have passed with no later visible failure, and the security
scan has cleared the patch with no later visible critical issue. Both
results must concern the same merge candidate as the current proposal.

\begin{figure}[t]
\centering
\begin{tikzpicture}[
  >=Latex,
  agent/.style={draw, rounded corners=2pt, fill=blue!5, minimum width=26mm, minimum height=9mm, align=center, font=\small},
  msg/.style={->, thick},
]
  \node[agent] (orch) at (0,   2.2) {Orchestrator};
  \node[agent] (comm) at (8.6, 2.2) {Committer};
  \node[agent] (test) at (4.3, 3.4)   {TestRunner};
  \node[agent] (sec)  at (4.3, 1)   {Security};

  \draw[msg] (orch) -- node[above, font=\scriptsize, pos=.5] {merge proposal} (comm);
  \draw[msg] (orch) -- node[above, sloped, font=\scriptsize, pos=.44] {patch} (test);
  \draw[msg] (orch) -- node[below, sloped, font=\scriptsize, pos=.48] {patch} (sec);
  \draw[msg] (test) -- node[above, sloped, font=\scriptsize, pos=.54] {test results~~~~~} (comm);
  \draw[msg] (sec)  -- node[below, sloped, font=\scriptsize, pos=.54] {scan results~~~~~} (comm);
\end{tikzpicture}
\caption{Architecture for concurrent code review. The Orchestrator distributes the patch to the TestRunner and Security agent, and sends a merge proposal to the Committer when a candidate is ready to be considered. The proposal triggers the Committer's decision procedure. The reviewers run in parallel and may send updated assessments after their first report.}
\label{fig:review-architecture}
\end{figure}

Figure~\ref{fig:review-architecture} shows the architecture. The merge
proposal is an ordinary message to the Committer. Its receive event stores
the proposed candidate locally and triggers the decision procedure. The guard is
evaluated at the following owner-side choice event. It checks the latest visible state of each reviewer, and also checks
that these states refer to the current candidate:
\[
\begin{aligned}
  &\AtField{\mathsf{TestRunner}}{\mathit{candidate}} = \mathit{candidate}
  ~\wedge~
  \AtField{\mathsf{Security}}{\mathit{candidate}} = \mathit{candidate}\\
  {}\wedge\;&
  \At{\mathsf{TestRunner}}{
    (\mathit{status}\neq\mathsf{failed}\wedge
     \mathit{status}\neq\mathsf{pending})
    \,\Since\, \mathit{status}=\mathsf{passed}
  } \\
  {}\wedge\;&
  \At{\mathsf{Security}}{
    (\mathit{status}\neq\mathsf{critical}\wedge
     \mathit{status}\neq\mathsf{pending})
    \,\Since\, \mathit{status}=\mathsf{cleared}
  }.
\end{aligned}
\]
Here \(\mathit{candidate}\) is the merge candidate stored at the current
Committer event. The term \(\AtField{A}{x}\) denotes the value of variable
\(x\) at the latest event of lifeline \(A\) that is causally visible to the
Committer. Thus, the two equalities check that the latest visible reviewer
episodes concern the current candidate. The formula \(\At{A}{\varphi}\) evaluates
\(\varphi\) at that same latest visible event of \(A\). The temporal subformulas, where \(\Since\) is the past-time since
operator, are evaluated along the local histories of the review lifelines.

Each reviewer processes candidate episodes sequentially. The TestRunner
status ranges over \(\mathsf{pending}\), \(\mathsf{checking}\),
\(\mathsf{passed}\), and \(\mathsf{failed}\).
When it starts candidate \(c\), one event atomically sets
\(\mathit{candidate}:=c\) and \(\mathit{status}:=\mathsf{pending}\). This
\(\mathsf{pending}\) event delimits the new episode. The value
\(\mathsf{passed}\) records that the required pre-merge checks have cleared,
and later \(\mathsf{checking}\) events may record additional testing. The
temporal condition remains true through such events. A later
\(\mathsf{failed}\) event invalidates the pass, while a later
\(\mathsf{pending}\) event begins another episode and prevents the earlier
pass from witnessing the condition. A new \(\mathsf{passed}\) event is then
required. Together, the candidate equality identifies the current episode
and the temporal condition checks its status history.

The Security lifeline analogously uses \(\mathsf{pending}\),
\(\mathsf{scanning}\), \(\mathsf{cleared}\), and \(\mathsf{critical}\), and
resets its candidate and status atomically at the start of each review.
Each reviewer lifeline has one current candidate. Sequential candidate
episodes share that lifeline, and simultaneous independent reviews use separate
reviewer lifelines or candidate-indexed application state.

The helper \texttt{review\_ok} in Listing~\ref{lst:impl} constructs the two
episode conditions above. The resulting guard directly controls the branch
owned by the Committer.

\begin{figure}[t]
\refstepcounter{cpllisting}
\label{lst:impl}
\noindent\textbf{Listing \thecpllisting.} Complete CPL guard and owner-side
choice in ZipperGen.
\begin{lstlisting}[language=ZipperGen,numbers=none,
  basicstyle=\ttfamily\footnotesize,columns=flexible,breaklines=true]
not_pending = Here.status != "pending"

def review_ok(A, bad, good):
    return At[A](since(
        (Here.status != bad) & not_pending,
        Here.status == good))

merge_guard = (
    (At[TestRunner].candidate == Here.candidate)
    & (At[Security].candidate == Here.candidate)
    & review_ok(TestRunner, "failed", "passed")
    & review_ok(Security, "critical", "cleared"))

@workflow
def merge_candidate(candidate: str @ Orchestrator) -> str:
    ...
    if merge_guard @ Committer:
        Committer: decision = approve(candidate)
    else:
        Committer: decision = reject(candidate)
\end{lstlisting}
\end{figure}

Reviewer reports make the corresponding reviewer states causally visible to
the Committer. The guard itself is stated over those reviewer states. A
Committer-local encoding would
have to reconstruct the reviewers' local orders from received reports, for
example using sequence numbers or causal metadata. This is especially
relevant when reports are propagated through intermediate participants or
delivered out of order: their arrival order at the Committer may differ from
the reviewer's local order.

\begin{figure}[t]
\centering
\begin{tikzpicture}[
  x=1.05cm,
  y=.92cm,
  >=Latex,
  font=\scriptsize,
  event/.style={circle, fill=black, inner sep=1.35pt},
  staleevent/.style={circle, fill=red!70!black, inner sep=1.35pt},
  stale/.style={red!70!black},
  line/.style={gray!65}
]
  \foreach \x/\name in {0/Orchestrator, 3.5/TestRunner, 7/Committer} {
    \node[font=\small] at (\x, 0) {\name};
    \draw[line] (\x, -0.35) -- (\x, -6.2);
  }

  \coordinate (tsend)  at (3.5, -1.0);
  \coordinate (trecv)  at (7,   -1.9);
  \coordinate (tfsend) at (3.5, -2);
  \coordinate (osend)  at (0,   -3.2);
  \coordinate (orecv)  at (7,   -4.1);
  \coordinate (cdec)   at (7,   -5.0);
  \coordinate (tfrecv) at (7,   -5.9);

  \foreach \p in {tsend, trecv, osend, orecv, cdec}
    \node[event] at (\p) {};
  \foreach \p in {tfsend, tfrecv}
    \node[staleevent] at (\p) {};

  \node[left=2mm, align=right] at (tsend)  {required check passes\\result sent};
  \node[left=2mm, align=right, stale] at (tfsend) {full check fails\\result sent};
  \node[left=2mm, align=right] at (osend)  {send merge proposal};
  \node[right=2mm, align=left]  at (cdec)  {choice event};
  \node[right=2mm, align=left, stale] at (tfrecv) {failure arrives\\too late};

  \draw[->] (tsend)  -- node[above, sloped, pos=.5] {\textsf{test results}}  (trecv);
  \draw[->] (osend)  -- node[above, sloped, pos=.3] {merge proposal}     (orecv);
  \draw[->, stale] (tfsend) -- node[above, sloped, pos=.3] {\textsf{test results}} (tfrecv);
\end{tikzpicture}
\caption{Causal structure at a merge decision. The Committer sees the test
  pass at the choice event, while the later failure remains in transit. The
  earlier assignment messages and local status-update events are omitted. The
  latest visible state on the omitted Security lifeline satisfies its guard
  conditions.}
\label{fig:review-trace}
\end{figure}

In Figure~\ref{fig:review-trace}, the required-check result has reached the Committer before
the merge proposal. A subsequent check has already found a
failure, but the corresponding update remains in transit. Therefore,
the latest TestRunner state visible to the Committer is still the
required-check pass. At the decision point, the formula
\(
  \At{\mathsf{TestRunner}}{
    (\mathit{status}\neq\mathsf{failed}\wedge
     \mathit{status}\neq\mathsf{pending})
    \,\Since\,\mathit{status}=\mathsf{passed}}
\)
holds, and the formula
\(
  \AtField{\mathsf{TestRunner}}{\mathit{candidate}}=\mathit{candidate}
\)
identifies the visible testing episode as belonging to the proposed
candidate. The \(\mathsf{pending}\) reset at the start of that episode keeps
the temporal witness within the same episode.

A monitor using a sequential log ordered by send time would see the failure
before the decision and block the merge. The Committer evaluates its guard
from the information causally visible at the decision point. Here, merging
becomes enabled when the required checks pass and remains enabled until a
later failure becomes visible.

\section{Message Sequence Charts}
\label{sec:executions}

We model agent-workflow executions as \emph{message sequence charts}
(MSCs). An MSC consists of \emph{events} distributed over lifelines. Each
event represents one observable step: sending or receiving a message,
executing a local action, or making a branching decision.
Events on the same lifeline are linearly ordered. Messages are represented
by matching send and receive events. Unmatched sends remain in transit in
prefixes. The causal order combines local successors with messages.
Figure~\ref{fig:review-trace} shows a small MSC for the
merge scenario: the Orchestrator sends a merge proposal to the Committer,
the TestRunner sends a result that arrives before the proposal, and a later
update is still in transit at the choice event.

We enrich each event with a \emph{valuation}. The valuation at an event
records the local state of the corresponding lifeline after that event. It
is a mapping from variables to values, for instance the current merge
candidate or the latest test status. The MSC semantics specifies where
formulas are evaluated and which valuations they may inspect. Agent
implementations determine their internal store semantics.

Let \(\Lifelines\) be a finite set of lifelines.

\begin{definition}[Message Sequence Chart (MSC)
  \leanformalized{https://github.com/zippergen-io/zippergen-lean/blob/main/icfem/CPLMonitor/MSC.lean\#L32}]
A \emph{Message Sequence Chart (MSC)} over \(\Lifelines\) is a tuple
\(M=(\Events,\to,\msg,\kind,\pid,\val)\), where \(\Events\) is a finite set of
events, \({\to} \subseteq \Events \times \Events\) is the local successor
relation, \({\msg} \subseteq \Events \times \Events\) is the message relation,
\(
  \kind:\Events\to
  \{\mathsf{act},\mathsf{recv},\mathsf{choice}\}
  \uplus(\{\mathsf{send}\}\times\Lifelines)
\)
classifies events.  A send event has kind \((\mathsf{send},B)\), where
\(B\) is its intended receiver.  The map \(\pid:\Events\to\Lifelines\)
assigns each event to its lifeline, and \(\val\) assigns an event valuation to each event. For
\(A\in\Lifelines\), define
\(
  \Events_A = \{e\in\Events \mid \pid(e)=A\}
\)
and
\(
  {\to_A}  = {\to} \cap (\Events_A\times\Events_A).
\)
The tuple satisfies the following conditions.
\begin{enumerate}[label=(\roman*)]
\item The relation \(\to\) is local: if \(e\to f\), then
\(\pid(e)=\pid(f)\).

\item For each \(A\), the relation \(\to_A\) is the immediate-successor
relation of a finite linear order on \(\Events_A\). We write \(\leq_A\)
for its reflexive transitive closure and \(<_A\) for the corresponding
strict order.

\item The message relation \(\msg\) is a partial matching between send and
receive events. If \(s\msg r\), then
\(\kind(s)=(\mathsf{send},\pid(r))\), \(\kind(r)=\mathsf{recv}\), and
\(\pid(s)\neq\pid(r)\).
Every receive has exactly one matching send, and every send has at most
one matching receive. Unmatched sends are messages still in transit, but
their intended receiver is still given by their event kind. Message delivery
may be non-FIFO.

\item The graph generated by local successor edges and message edges is
acyclic.
\end{enumerate}
\end{definition}

Unless stated otherwise, let \(M\) be an MSC over \(\Lifelines\).

The Lean development formalizes event kinds, message matching, and the
semantic interface used by CPL: lifelines, local order, causal order,
valuations, and the previous-local and latest-visible navigation
operations.

The causal order \(\causal_M\) is the reflexive transitive closure of local
successor edges and message edges.
By acyclicity, \(\causal_M\) is a partial order.
We write \(e \strictcausal_M f\) if \(e \causal_M f\) and \(e\neq f\).

For an event $e$, we write \(\nu_e := \val(e)\).
Thus \(\nu_e(x)\) is the value of variable \(x\) at event \(e\).
One can think of \(\nu_e\) as the local store of the lifeline after event
\(e\). In ZipperGen, actions and receives may update this store, while send
and choice events leave it unchanged. For choice events, \(\nu_e\) is the store on which the
guard is evaluated.

\section{Causal Past Logic}
\label{sec:cpl}

\begin{definition}[CPL syntax
  \leanformalized{https://github.com/zippergen-io/zippergen-lean/blob/main/icfem/CPLMonitor/CPLSyntax.lean\#L11}]
\label{def:cpl-syntax}
\emph{Variable terms} $t$, \emph{atomic conditions} $\alpha$, and \emph{CPL formulas} $\varphi$
are given by the following grammar:
\[
  \begin{array}{rcl}
  t &::=& x \mid \AtField{A}{x} \qquad
  \alpha ::= p(t_1,\ldots,t_n) \\[0.5ex]
  \varphi &::=& \alpha
    \mid \At{A}{\varphi}
    \mid \Prev\varphi
    \mid \varphi_1 \Since \varphi_2 \mid \varphi_1\wedge\varphi_2
    \mid \varphi_1\vee\varphi_2
    \mid \neg\varphi .
  \end{array}
\]
\end{definition}
Here, \(x\) ranges over workflow variables and \(p\) ranges over
application-level predicate symbols.  We use standard mathematical
notation for common predicates.  Thus \(x<10\) abbreviates a unary
predicate on \(x\), \(\mathit{status}=\mathsf{approved}\) abbreviates a
unary predicate on \(\mathit{status}\), and
\(\AtField{\mathsf{TestRunner}}{\mathit{candidate}}=\mathit{candidate}\)
abbreviates a binary equality predicate applied to the two displayed
terms. Fixing any CPL formula \(\chi\), we use the Boolean abbreviations
\(\true := \chi\vee\neg\chi\) and
\(\false := \chi\wedge\neg\chi\).

Moreover, \(\Prev\) and \(\Since\) are the usual previous and since operators
from past-time linear temporal logic, interpreted along the current
lifeline. The operator \(\At{A}{\varphi}\) is the causal part: it evaluates
\(\varphi\) at the latest visible event of lifeline \(A\). The term
\(\AtField{A}{x}\) is different. It reads the value of \(x\) at that event
and can only occur inside atomic conditions.

Let \(e\in\Events\), and let \(A=\pid(e)\). Define the previous local
event, if it exists, by
\(
  \lastloc(e) = \max \{f\in\Events_A \mid f <_A e\}.
\)
For any lifeline \(B\), define the latest visible event of \(B\) from \(e\)
by
\(
  \last_B(e) = \max \{f\in\Events_B \mid f \causal_M e\}.
\)
Both maxima are unique when they exist, since events on a lifeline are
linearly ordered.

For a term \(t\), let \(\sem{t}_M(e)\) denote its value when evaluated at
event \(e\), if defined. Since \(\nu_e\) records the local store after \(e\) and
is total,
\(
  \sem{x}_M(e)=\nu_e(x)
\)
is always defined. If \(\last_A(e)\) exists, then
\(\sem{\AtField{A}{x}}_M(e)=\nu_{\last_A(e)}(x)\). Otherwise, the
cross-lifeline term is undefined. If any term in a
predicate atom or comparison is undefined, then the atom is false. In particular, if no latest visible \(A\)-event exists,
then any atomic condition using \(\AtField{A}{x}\) is false. Thus
\(\AtField{A}{x}\) always reads the latest visible value of \(x\) on
lifeline \(A\). This is the current event itself when \(e\) is on \(A\),
and the latest causally visible \(A\)-event otherwise. Strict movement is
expressed separately with the local previous operator \(\Prev\). For
example, \(\At{A}{\Prev\varphi}\) moves to the latest visible \(A\)-event
and then one local step back.

The operators \(\At{A}{\varphi}\) and \(\AtField{A}{x}\) correspond to the
remote formulas and remote expressions of PT-DTL~\cite{SenVAR04}, adapted
to MSC events. PT-DTL assumes that initial remote values are known
throughout the system. Here, a remote term with no visible event is
undefined and makes its atom false.

The application fixes the interpretation of each predicate symbol \(p\) as
a relation on data values. The term semantics above supplies its arguments.
In particular, an unqualified variable is read from \(\nu_e\), while an
\(\AtField{A}{x}\) term is resolved through the causal structure of the MSC.

\begin{definition}[CPL semantics
  \leanformalized{https://github.com/zippergen-io/zippergen-lean/blob/main/icfem/CPLMonitor/CPLSemantics.lean\#L44}]
\label{def:cpl-semantics}
The satisfaction relation \(M,e\models\varphi\), defined for
\(e\in\Events\), is given inductively as follows. For an atomic condition
\(\alpha=p(t_1,\ldots,t_n)\), we have \(M,e\models\alpha\) iff all
\(v_i=\sem{t_i}_M(e)\) are defined and \(p(v_1,\ldots,v_n)\) holds.
The previous and latest-visible cases are
\begin{align*}
M,e &\models \Prev\varphi
  &&\text{iff } \lastloc(e) \text{ is defined and }
    M,\lastloc(e)\models\varphi,\\
M,e &\models \At{A}{\varphi}
  &&\text{iff } \last_A(e) \text{ is defined and }
    M,\last_A(e)\models\varphi.
\end{align*}
Boolean connectives are interpreted as usual. If \(\lastloc(e)\) is
undefined, then \(\Prev\varphi\) is false. If no latest visible
\(A\)-event exists, then \(\At{A}{\varphi}\) is false. For \(e\in\Events_A\),
we have \(M,e\models \varphi_1\Since\varphi_2\) iff there exists
\(f\in\Events_A\) with \(f\leq_A e\),
\(M,f\models\varphi_2\), and \(M,g\models\varphi_1\) for every
\(g\in\Events_A\) with \(f<_A g\leq_A e\).
\end{definition}

Thus the \(\Since\) navigation is local, while \(\varphi_1\) and \(\varphi_2\) may
themselves contain causal modalities. For example,
\((\mathit{status}\neq\mathsf{failed}\wedge
\mathit{status}\neq\mathsf{pending})\Since
\mathit{status}=\mathsf{passed}\) holds at a TestRunner event \(e\) if the
current candidate episode contains a pass followed only by
\(\mathsf{checking}\) or \(\mathsf{passed}\) events, up to \(e\).

Causal past is a derived notation and is non-strict:
$\Past\varphi
  :=
  \bigvee_{B\in\Lifelines}
    \At{B}{\true\Since\varphi}$.
The lifeline-specific form is
$\PastAt{A}{\varphi}
  :=
  \At{A}{\true\Since\varphi}$.

Thus, \(\PastAt{A}{\varphi}\) holds at \(e\) exactly when some event on
lifeline \(A\) visible from \(e\) satisfies \(\varphi\), and
\(\Past\varphi\) is the disjunction over all lifelines.
For instance, the formula
\(
  \PastAt{\mathsf{TestRunner}}{\mathit{status}=\mathsf{passed}}
\)
holds at an event \(e\) if some event on the TestRunner lifeline visible
from \(e\) satisfies \(\mathit{status}=\mathsf{passed}\).

The abbreviation
\(
  \Seen{A} := \At{A}{\true}
\)
tests whether some event of \(A\) is visible. If the current event itself
is on \(A\), this is true by non-strictness.

\section{Online Monitoring}
\label{sec:monitoring}

At runtime, each guard is evaluated by its owner lifeline. The monitor must therefore track which events are causally visible at the current event. Each lifeline \(A\) maintains a vector clock~\cite{Lamport78,fidge1988timestamps,mattern1989virtual}, which is of the form
\(
  \vc_A : \Lifelines \to \mathbb N.
\)
We use vector clocks after the current event has been processed. Thus,
after processing an event \(e\in\Events_A\), the value \(\vc_A(B)\)
records how many events of \(B\) are in the causal past of
\(e\), including \(e\) itself when \(B=A\). This matches the non-strict
semantics of \(\At{B}{\varphi}\) and \(\AtField{B}{x}\): both refer to
the latest visible \(B\)-event. We let \(e^B_k\) denote the \(k\)-th event on lifeline \(B\)
(where indexing starts at 1).

Let \(\Phi\) be a finite set of CPL formulas, \(\sub(\Phi)\) the set of
all their subformulas, and \(\vars(\Phi)\) the set of variable
names that appear in \emph{cross-lifeline} terms \(\AtField{B}{x}\) in
\(\Phi\). For the workflow under consideration, \(\Phi\) contains
all CPL guards that can be evaluated at runtime. All lifelines use the
same \(\Phi\). Derived shorthands
such as \(\Past_A\) are unfolded to primitive CPL syntax before taking
subformulas. The monitor state is designed to answer three kinds of local
queries: strict previous-local formulas, latest visible formulas on each
lifeline, and latest visible field values used by terms
\(\AtField{B}{x}\). The state follows the knowledge-vector scheme
of~\cite{SenVAR04}. We use full vector clocks, which also serve as event
indices and presence tests in the correctness proof. This is formalized in
the following definition:

\begin{definition}[Monitor state
  \leanformalized{https://github.com/zippergen-io/zippergen-lean/blob/main/icfem/CPLMonitor/MonitorState.lean\#L21}]
\label{def:monitor-state}
At each lifeline \(A\), the monitor maintains:
\begin{itemize}
\item a \emph{vector clock} \(\vc_A : \Lifelines \to \mathbb{N}\).
\item a partial \emph{Boolean latest-value view}
  \(\view_A : \Lifelines\times\sub(\Phi)\rightharpoonup\{0,1\}\).
\item a partial \emph{variable view}
  \(\var_A : \Lifelines\times\vars(\Phi)\rightharpoonup\mathsf{Val}\),
  where \(\mathsf{Val}\) is the domain of stored values.
\item a \emph{local store} \(\sigma_A\) for the current event valuation.
\end{itemize}
It also uses a temporary \emph{previous-local copy} during an event update:
\(\old_A(\psi)\in\{0,1\}\) for each \(\psi\in\sub(\Phi)\).
In a stable state, after an event update has completed, the clock is the
presence test for the two latest-value tables: entries with first
component \(B\) are defined exactly when \(\vc_A(B)>0\). If
\(\vc_A(B)=k>0\), those entries refer to the \(k\)-th event of~\(B\). If
\(\vc_A(B)=0\), they are absent.
\end{definition}

The temporary copy \(\old_A\) is total. At the first
event of \(A\), it is set to \(0\) for every \(\psi\in\sub(\Phi)\).

For a fixed workflow, the Boolean part of the monitor state has size \(O(|\Lifelines|\cdot|\sub(\Phi)|)\) per lifeline, including the previous-local copy \(\old_A\). Variable views add \(O(|\Lifelines|\cdot|\vars(\Phi)|)\) selected data values. The vector-clock components are unbounded counters. In PT-DTL monitoring, the freshness vector contains entries only for processes named by monitored remote expressions and formulas~\cite{SenVAR04}. The same optimization applies here. The unbounded counters agree with the impossibility of deterministic finite-state gossip in unbounded message-passing systems~\cite{BFG2018}. In practice, they are bounded by the lifetime of the run or a timeout.

We reason about a run by fixing an arbitrary linear extension of the causal order of the MSC. Each event is processed by the monitor of its own lifeline when that event occurs. The prefix processed before an event is therefore causally closed: if an event has been processed, all of its causal predecessors have been processed as well. Initially, all clock components are zero and the latest-value tables are empty. We assume that applying \(\mathsf{effect}\) at an event \(e\) produces the MSC valuation \(\val(e)\), and that a receive obtains the monitor snapshot emitted at its matching send.
On each lifeline, the state produced at one event is used at the next.
The first event starts from the initial monitor state.

When lifeline \(A\) creates a new event, the monitor runs Algorithm~\ref{alg:event-update}. For a receive, it first merges the vector clock, Boolean views, and variable views carried by the incoming message. Entries are copied from the message exactly when the message is ahead on a lifeline \(B\). In that case the message clock is positive on \(B\), so the corresponding message entries are defined. The monitor then increments its local clock component, updates the local store, evaluates all subformulas (Algorithm~\ref{alg:formula-eval}), and records the results. All messages carry vector clocks and latest-value views. The algorithm treats user messages and projection-internal control messages uniformly.

The function \(\mathsf{effect}\) accounts for the event's local store update.
In the current ZipperGen runtime, actions and receives may update the store,
while sends and choice events use the identity update.

\begin{lstlisting}[float=!t,language=Monitor,mathescape=true,
  caption={Event update and view propagation.},label={alg:event-update}]
procedure ON_EVENT($A,e$)
  if $\kind(e)=\mathsf{recv}$ then
    $\mu \gets \mathsf{message}(e)$
    for $B\in\Lifelines$ do
      if $\mu.\vc(B)>\vc_A(B)$ then
        for $\psi\in\sub(\Phi)$ do
          $\view_A(B,\psi) \gets \mu.\view(B,\psi)$
        end
        for $x\in\vars(\Phi)$ do
          $\var_A(B,x) \gets \mu.\var(B,x)$
        end
      end
    end
    $\vc_A \gets \max(\vc_A,\mu.\vc)$
  end

  for $\psi\in\sub(\Phi)$ do
    $\old_A(\psi) \gets \view_A(A,\psi)$ if defined, otherwise $0$
  end

  $\vc_A(A) \gets \vc_A(A)+1$
  $\sigma_A \gets \mathsf{effect}(e,\sigma_A)$
  for $x\in\vars(\Phi)$ do
    $\var_A(A,x) \gets \sigma_A(x)$
  end

  for $\psi\in\sub(\Phi)$ do
    $\mathsf{res}_A(\psi) \gets \Eval_A(\psi,e)$
  end
  for $\psi\in\sub(\Phi)$ do
    $\view_A(A,\psi) \gets \mathsf{res}_A(\psi)$
  end

  if $\kind(e)=(\mathsf{send},B)$ then
    emit to $B$ the message $(\mathsf{payload}(e),\vc_A,\view_A,\var_A)$
  end
end
\end{lstlisting}

Algorithm~\ref{alg:formula-eval} computes truth values by structural recursion,
following the standard recurrence scheme for past-time
operators~\cite{HavelundR02}. During an
event update, \(\Eval_A\) reads the current monitor state of \(A\), including
the previous-local snapshot \(\old_A\). If there
is no previous local event, \(\Prev\theta\) is false. If no latest visible
\(B\)-event exists, \(\At{B}{\theta}\) is false. The \(\At{B}{\theta}\) case is non-strict: for \(B=A\) it evaluates \(\theta\) at the current event, and for \(B\neq A\), it first checks whether a \(B\)-event is visible and only then reads \(\view_A(B,\theta)\). A previous \emph{local} event exists when \(\vc_A(A)>1\). We write \((\sigma_A,\var_A)\models\alpha\) for runtime atomic evaluation: unqualified variables are read from \(\sigma_A\), terms \(\AtField{B}{x}\) are read from \(\var_A(B,x)\), and the resulting values are passed to the same application predicate as in the MSC semantics. A missing required field makes the atom false. For example, the Committer evaluates
\(\AtField{\mathsf{TestRunner}}{\mathit{candidate}}=\mathit{candidate}\) by
comparing \(\var_{\mathsf{Committer}}(\mathsf{TestRunner},\mathit{candidate})\)
with its current local candidate.

Suppose \(A\neq B\), and let \(b\) be the latest \(B\)-event visible at
the current event of \(A\). In \(\At{B}{\At{A}{\varphi}}\), the inner
formula refers to the latest \(A\)-event visible from \(b\), if any.
The monitor therefore reads \(\view_A(B,\At{A}{\varphi})\), rather than
evaluating \(\varphi\) at the current \(A\)-event. For
\(\Prev\At{B}{\varphi}\), it reads \(\old_A(\At{B}{\varphi})\), the
value recorded at the previous local event. Newly received information
must not be used to recompute that value.

The Lean evaluator uses \texttt{none} for an absent view and initializes
the previous-local table to false. For states produced by
Algorithm~\ref{alg:event-update}, these conventions are equivalent to the
explicit clock tests in Algorithm~\ref{alg:formula-eval}.

\begin{lstlisting}[float=!t,language=Monitor,mathescape=true,
  caption={Formula evaluation at the current event on lifeline \(A\).},label={alg:formula-eval}]
procedure $\Eval_A(\psi,e)$
  case $\psi$ of
    $\alpha$: return $(\sigma_A,\var_A)\models \alpha$
    $\Prev\theta$:
      return $(\vc_A(A)>1)\wedge\old_A(\theta)$
    $\At{B}{\theta}$:
      if $B=A$ then return $\Eval_A(\theta,e)$
      else if $\vc_A(B)=0$ then return $0$
      else return $\view_A(B,\theta)$
    $\theta_1\Since\theta_2$:
      $\mathit{prev} \gets (\vc_A(A)>1)\wedge\old_A(\theta_1\Since\theta_2)$
      return $\Eval_A(\theta_2,e)\vee(\Eval_A(\theta_1,e)\wedge \mathit{prev})$
    $\neg\theta$: return $\neg\Eval_A(\theta,e)$
    $\theta_1\wedge\theta_2$: return $\Eval_A(\theta_1,e)\wedge\Eval_A(\theta_2,e)$
    $\theta_1\vee\theta_2$: return $\Eval_A(\theta_1,e)\vee\Eval_A(\theta_2,e)$
  end
end
\end{lstlisting}

\begin{definition}[Coherent pre-evaluation state \leanformalized{https://github.com/zippergen-io/zippergen-lean/blob/main/icfem/CPLMonitor/LocalEval.lean\#L29}]
\label{def:coherent-pre-evaluation}
Let \(e\in\Events_A\). The monitor state of \(A\) is
\emph{coherent before evaluating formulas at \(e\)} if the following
conditions hold.
\begin{enumerate}[label=(\roman*)]
\item For every lifeline \(B\),
      \[
        \vc_A(B)=|\{f\in\Events_B \mid f\causal_M e\}|.
      \]
      Thus, if \(\vc_A(B)=k>0\), then the latest visible event of \(B\)
      from \(e\) is \(e^B_k\).
\item For every \(B\neq A\) with \(k=\vc_A(B)>0\), the entries for
      \(B\) are defined and describe \(e^B_k\):
      \[
        \view_A(B,\psi)=1
        \quad\text{iff}\quad
        M,e^B_k\models\psi
      \]
      for all \(\psi\in\sub(\Phi)\), and
      \[
        \var_A(B,x)=\val(e^B_k)(x)
      \]
      for all \(x\in\vars(\Phi)\).  If \(\vc_A(B)=0\), then
      no \(B\)-event is visible and all entries
      \(\view_A(B,\psi)\) and \(\var_A(B,x)\) are absent.
\item For the owner lifeline \(A\), the local store induces the current
      event valuation: for every
      variable \(x\),
      \[
        \sigma_A(x)=\val(e)(x).
      \]
      Moreover,
      \(\var_A(A,x)\) is defined and equal to \(\val(e)(x)\) for every
      \(x\in\vars(\Phi)\).
      No condition is imposed on the owner Boolean entries
      \(\view_A(A,\psi)\) before formula evaluation. They are overwritten
      by Algorithm~\ref{alg:event-update} after the subformulas have been
      evaluated.
\item For every \(\psi\in\sub(\Phi)\), \(\old_A(\psi)\) is the truth
      value of \(\psi\) at the previous
      local event of \(A\), if such an event exists, and is false
      otherwise.
\end{enumerate}
\end{definition}

\begin{lemma}[Local evaluation \leanproof{https://github.com/zippergen-io/zippergen-lean/blob/main/icfem/CPLMonitor/LocalEval.lean\#L111}]
\label{lem:local-evaluation}
Let \(e\in\Events_A\). Suppose the monitor state of \(A\) is coherent
before evaluating formulas at \(e\). Then, for every
\(\psi\in\sub(\Phi)\),
\[
  \Eval_A(\psi,e)=1
  \quad\text{iff}\quad
  M,e\models\psi .
\]
\end{lemma}

\begin{proof}
The proof is by structural induction on \(\psi\).

\smallskip\noindent\textit{Atomic condition \(\alpha\).}
The runtime evaluates \(\alpha\) against the current store \(\sigma_A\)
and the latest-visible variable views \(\var_A(B,x)\).  Unqualified
variables are read from \(\sigma_A\), which agrees with \(\val(e)\) by
coherence.  For \(\AtField{A}{x}\), coherence gives the value of \(x\)
at \(\last_A(e)=e\).  For \(\AtField{B}{x}\) with \(B\neq A\),
coherence gives the value of \(x\) at \(e^B_{\vc_A(B)}=\last_B(e)\) when
a visible \(B\)-event exists. If none exists, both the runtime and the
semantics make the atom false. Thus every term has the same definedness and
value in both evaluations. Both apply the same interpretation of the
predicate symbol, so their truth values agree.

\smallskip\noindent\textit{Boolean connectives.}
The cases for \(\neg\theta\), \(\theta_1\wedge\theta_2\), and
\(\theta_1\vee\theta_2\) follow directly from the induction
hypothesis.

\smallskip\noindent\textit{Previous: \(\Prev\theta\).}
If \(\vc_A(A)=1\), then \(e\) is the first local event of \(A\), so
there is no previous local event and both sides are false.  Otherwise,
coherence says that \(\old_A(\theta)\) is the truth value of \(\theta\)
at the previous local event, which is exactly the semantics of \(\Prev\).

\smallskip\noindent\textit{At-operator: \(\At{B}{\theta}\).}
If \(B=A\), non-strictness gives \(\last_A(e)=e\), and the algorithm
recursively evaluates \(\theta\) at the current event. The result is
correct by induction. If \(B\neq A\), the algorithm reads the
latest-value view \(\view_A(B,\theta)\), which is correct by coherence.
If no \(B\)-event is visible, both sides are false.

\smallskip\noindent\textit{Since: \(\theta_1\Since\theta_2\).}
The algorithm uses the standard recurrence for past-time since over the
local order of \(A\), i.e.,
\(
  \theta_1\Since\theta_2
  \equiv
  \theta_2 \vee \bigl(\theta_1\wedge \Prev(\theta_1\Since\theta_2)\bigr)
\).
The current values of \(\theta_1\) and \(\theta_2\) are correct by the
induction hypothesis, and the previous value of the since formula is
provided by \(\old_A\).
\qed
\end{proof}

\begin{lemma}[Monitor-state invariants \leanproof{https://github.com/zippergen-io/zippergen-lean/blob/main/icfem/CPLMonitor/EventUpdate.lean\#L1067}]
\label{lem:monitor-invariants}
After Algorithm~\ref{alg:event-update} has processed an event \(e\in\Events_A\),
the following hold for every lifeline \(B\):
\begin{enumerate}[label=(\roman*)]
\item \(\vc_A(B) = |\{f\in\Events_B \mid f \causal_M e\}|\).
\item If \(\vc_A(B)=0\), then all entries
  \(\var_A(B,x)\) and \(\view_A(B,\psi)\) are absent.
\item If \(\vc_A(B)=k>0\), then \(\var_A(B,x)\) is defined and equal to
  \(\val(e^B_k)(x)\), for every \(x\in\vars(\Phi)\).
\item If \(\vc_A(B)=k>0\), then \(\view_A(B,\psi)\) is defined and
  \[
    \view_A(B,\psi)=1
    \quad\text{iff}\quad
    M,e^B_k\models\psi
  \]
  for every \(\psi\in\sub(\Phi)\).
\end{enumerate}
\end{lemma}
\begin{proof}
Fix a linear extension of the causal order. We proceed by induction over this order.
Equivalently, the induction hypothesis is available for the causally
closed prefix processed before the current event.

\smallskip\noindent\textit{Base case.}
Before any event is processed, \(\vc_A(B)=0\) for all \(B\).
No event is yet causally visible, so (i) holds.
The latest-value tables are empty, so (ii) holds. Conditions (iii) and (iv) are
irrelevant since no clock component is positive.

\smallskip\noindent\textit{Non-receive event \(e\) on \(A\).}
This covers act, send, and choice events.
The algorithm increments \(\vc_A(A)\) by one and leaves all other
clock components unchanged.

\begin{itemize}
\item[(i)] This matches the causal past: since \(e\) is not a receive,
  the only newly visible event is \(e\) itself, on lifeline \(A\).
\item[(ii)] No entry becomes absent when processing an event. For
  lifelines whose clock remains zero, the tables remain empty.

\item[(iii)] The algorithm writes \(\var_A(A,x)\gets\val(e)(x)\) from the
  current local store, which equals \(\val(e^A_{\vc_A(A)})(x)\) since
  \(e^A_{\vc_A(A)}=e\).
  For \(B\neq A\), \(\var_A(B,x)\) is unchanged and no new \(B\)-event became
  visible, so the stored value still refers to the correct event by induction.
\item[(iv)] Algorithm~\ref{alg:event-update}
  overwrites \(\view_A(A,\psi)\) with the value of \(\psi\) at \(e\).
  The monitor state is coherent
  before formula evaluation: (i) gives the clock condition, remote views are
  unchanged and correct by induction, the local store has just been updated
  to induce \(\val(e)\), local variable views have just been updated to
  \(\val(e)\), and \(\old_A\) contains the previous local truth values.
  Hence Lemma~\ref{lem:local-evaluation} gives the correctness of the new
  local Boolean entries.
  For \(B\neq A\), \(\view_A(B,\psi)\) is unchanged.
\end{itemize}

\smallskip\noindent\textit{Receive event \(e\) on \(A\), matching send \(s\) on \(B\).}
The message carries \(\vc_B^s\), \(\view_B^s\), \(\var_B^s\), the monitor
state of \(B\) immediately after \(s\).
By induction, all four invariants hold for \(B\) at \(s\).
Since \(s\strictcausal_M e\), every event causally visible to \(B\) at \(s\)
is also causally visible to \(A\) at \(e\).
Write \(\vc_A^{-}\), \(\view_A^{-}\), and \(\var_A^{-}\) for the monitor
state on \(A\) just before the receive merge.

The algorithm sets \(\vc_A(C)\gets\max(\vc_A(C),\vc_B^s(C))\) for each \(C\),
then increments \(\vc_A(A)\).

\begin{itemize}
\item[(i)] For \(C\neq A\): the only new non-local causal information at
  the receive event comes through the matching message edge \(s\msg e\).
  Hence a \(C\)-event visible at \(e\) is either already visible to \(A\)
  before the receive, or visible to \(B\) at the send event \(s\).
  By induction, the two counts are \(\vc_A^{-}(C)\) and \(\vc_B^s(C)\).
  Both sets of \(C\)-events are prefixes of the local order on \(C\).
  Therefore their union is the longer prefix, whose size is the maximum
  of the two sizes.

  For \(C=A\): any \(A\)-event visible to \(B\) at \(s\) is already a strict
  local predecessor of \(e\). Indeed, if \(e\leq_A a\), then
  \(a\causal_M s\strictcausal_M e\causal_M a\), a causal cycle. Hence
  \(\vc_B^s(A)\leq\vc_A^{-}(A)\).  The merge therefore
  does not change the \(A\)-component, and the final increment accounts for
  \(e\) itself.
\item[(ii)] If the resulting clock on \(C\) is zero, then both incoming
  snapshots have clock zero on \(C\), so both have absent entries by
  induction and the result remains absent.

\item[(iii)] The algorithm copies \(\var_B^s(C,x)\) into \(\var_A(C,x)\) exactly
  when \(\vc_B^s(C)>\vc_A^{-}(C)\), i.e., when the sender is strictly ahead on \(C\).
  By induction, we have \(\var_B^s(C,x)=\val(e^C_{\vc_B^s(C)})(x)\), so the copy is
  correct.
  When \(\vc_B^s(C)\leq\vc_A^{-}(C)\), the local view is already at least as
  recent as the sender's view, so no copy is needed.
  The local entry \(\var_A(A,x)\) is set to \(\val(e)(x)\) as in the
  non-receive case.
\item[(iv)] The same merge applies to \(\view_A(C,\psi)\): when
  \(\vc_B^s(C)>\vc_A^{-}(C)\), \(\view_B^s(C,\psi)\) is copied.
  By induction, we have $\view_B^s(C,\psi)=1$ iff $M,e^C_{\vc_B^s(C)}\models\psi$.
  For \(C\neq A\), after the vector-clock merge, the stored table entry
  refers to the correct latest visible event of \(C\).
  The inequality \(\vc_B^s(A)\leq\vc_A^{-}(A)\) established in (i) also
  ensures that the merge leaves the owner entries \(\view_A(A,\psi)\)
  unchanged. Hence the subsequent assignment to \(\old_A\) records the
  truth values at the previous local event.
  After the merge and the local updates of Algorithm~\ref{alg:event-update},
  the monitor state is coherent, so Lemma~\ref{lem:local-evaluation}
  applies as in the non-receive case.
\end{itemize}
This concludes the proof of Lemma~\ref{lem:monitor-invariants}.
\qed
\end{proof}

Suppose two messages from the same sender reach a receiver in reverse send
order. Every event visible at the earlier send is also visible at the later
send. By Lemma~\ref{lem:monitor-invariants}(i), the later snapshot therefore
has at least the same value in every clock component. Once it has been
merged, the older snapshot cannot replace any remote row. The older message
still causes a local receive event: the monitor increments its own clock,
applies the local store update, and evaluates the formulas.

For an event \(e\in\Events_A\), let \(\mathsf{res}_A^e(\varphi)\)
denote the value computed for \(\varphi\) by
Algorithm~\ref{alg:event-update} during the update for \(e\).
The theorem applies to individual executions. Agreement on branch outcomes
at the workflow level follows from ZipperGen's projection discipline.

\begin{theorem}[Monitor correctness \leanproof{https://github.com/zippergen-io/zippergen-lean/blob/main/icfem/CPLMonitor/EventUpdate.lean\#L1133}]
Let \(M\) be an MSC, fix any linear extension of its causal order, and
run the distributed monitors along this order. If \(e\in\Events_A\) is
an event on lifeline \(A\), then immediately after the monitor on \(A\)
has processed \(e\), every formula \(\varphi\in\Phi\) satisfies
\[
  \mathsf{res}_A^e(\varphi) = 1
  \quad\text{iff}\quad
  M,e \models \varphi.
\]
\end{theorem}

\begin{proof}
After processing \(e\), the algorithm has copied
\(\mathsf{res}_A(\varphi)\) into \(\view_A(A,\varphi)\).  By
Lemma~\ref{lem:monitor-invariants}(i), \(\vc_A(A)\) is the local index of
\(e\), so \(e^A_{\vc_A(A)}=e\).  Applying
Lemma~\ref{lem:monitor-invariants}(iv) with \(B=A\), we get
\(
  \view_A(A,\varphi)=1
  \text{ iff }
  M,e\models\varphi .
\)
Since \(\mathsf{res}_A^e(\varphi)=\view_A(A,\varphi)\) immediately after
the update, the claim follows.
\qed
\end{proof}

\section{Implementation and Preliminary Evaluation}
\label{sec:evaluation}

\paragraph{Implementation.}
CPL guards are implemented in a ZipperGen prototype\footnote{\url{https://zippergen.io}},
an open-source Python framework for multi-agent LLM coordination.
Guards are Python expressions using \texttt{At[B].x} for variable
terms, \texttt{Here.x} for the store at the current evaluation point, \texttt{since} for the
past-time condition, and standard Boolean connectives. The prototype keeps
the monitor state locally, sends vector clocks and latest-value views with
each message, and evaluates guards at owner-side choice events, before
executing the selected continuation. It retains the fields that occur in
cross-lifeline terms and evaluates each registered formula object once per
event in bottom-up order. For the benchmark formulas, the registered object
counts coincide with the set-based subformula counts. Listing~\ref{lst:impl} shows the
complete code-review guard and its use in a ZipperGen branch.

\paragraph{Monitor microbenchmark.}

We measured the time and message data used by the ZipperGen monitor.\footnote{Artifacts
and exact ZipperGen revisions: \href{https://github.com/zippergen-io/zippergen/tree/46b41e4c49ee91b734794a51300331813de70744/experiments/icfem-2026-cpl}{public repository}.}
The first configuration uses the four lifelines and the complete merge guard from
Section~\ref{sec:overview}: 16 subformulas and
\(\vars(\Phi)=\{\mathit{candidate}\}\).  The status tests use formula views.
Only \(\mathit{candidate}\) requires a variable view.
For each local status test, the benchmark reads \texttt{status} directly from
the local store instead of using the \texttt{Here.status} comparison shown in
Listing~\ref{lst:impl}. Because the local store is total, the two forms give
the same result, but they need not take the same time.

For the scaling cases, with \(1 \le k \le n\), a monitor on \(L_{\ell-1}\) stores view entries for
\(L_0,\ldots,L_{\ell-1}\) and evaluates
\(\Psi_{n,k}=\bigwedge_{i=0}^{n-1}
  \bigl(\AtField{L_0}{f_{i\bmod k}}=i\bigr)\).

The conjunction follows the structure of the merge guard, whose checks must all
hold.  Using disjunction would give essentially the same timings because the
monitor visits every subformula.  The field names repeat cyclically, while the
constants distinguish the atoms.  Hence \(\Psi_{n,k}\) has \(n\) atoms,
\(2n-1\) subformulas, and \(k\) variable names.  The case \(k=128\)
uses \(n=128\).

The three operations are timed separately. Before each measurement, the monitor
contains a stored value for every formula and tracked variable on every
lifeline. The local case processes an action event. The send case processes
a send event and copies the vector clock and the two views that the runtime
would attach to an outgoing message. The receive case first advances the synthetic
incoming clock entries for \(L_0,\ldots,L_{\ell-2}\), then passes the metadata to
\(L_{\ell-1}\). Its views contain the corresponding rows, so the monitor replaces
every remote row. The receive time includes these clock updates.
The timings cover only the monitor updates described above. Encoding and
transmitting messages, storing runtime state, and processing application
payloads are not timed.

For each operation, we discard 50 warm-up events and report the median of the
mean event times in nine batches of 300. Garbage collection is disabled during
measurement. Message size is the UTF-8 length of the JSON-encoded vector clock
and views, with one KB equal to 1,000 bytes. The experiments used CPython 3.12
on an Intel Core i9-9880H running macOS. The timings were stable: the
interquartile range for every recorded row was below 2.7\% of its median.

\begin{table}[t]
\caption{Time per monitor event and monitor data per message.}
\label{tab:microbenchmark}
\centering
\scriptsize
\renewcommand{\arraystretch}{1.08}
\setlength{\tabcolsep}{2.8pt}
\setlength{\aboverulesep}{.2ex}
\setlength{\belowrulesep}{.2ex}
\begin{tabular}{@{}lrrrrrrr@{}}
\toprule
Workload & $\ell$ & \shortstack{Atoms\\$n$}
 & \shortstack{Variables\\$k$} & \shortstack{Local\\(ms)}
 & \shortstack{Send\\(ms)} & \shortstack{Receive\\(ms)}
 & \shortstack{Message data\\(KB)}\\
\midrule
Code review & 4 & -- & 1 & 0.047 & 0.069 & 0.072 & 1.2\\
Scaling $\Psi_{n,k}$ & 2 & 16 & 4 & 0.135 & 0.158 & 0.154 & 1.1\\
Scaling $\Psi_{n,k}$ & 32 & 16 & 4 & 0.209 & 0.506 & 0.557 & 17.0\\
Scaling $\Psi_{n,k}$ & 8 & 4 & 4 & 0.057 & 0.095 & 0.102 & 1.9\\
Scaling $\Psi_{n,k}$ & 8 & 64 & 4 & 0.521 & 0.754 & 0.764 & 14.1\\
Scaling $\Psi_{n,k}$ & 8 & 128 & 1 & 0.998 & 1.430 & 1.448 & 27.3\\
Scaling $\Psi_{n,k}$ & 8 & 128 & 128 & 1.976 & 2.979 & 2.903 & 58.0\\
\bottomrule
\end{tabular}
\end{table}

The code-review guard takes about 0.05--0.07 ms per event and carries 1.2 KB.
Receive time rises from 0.10 to 0.76 ms as \(n\) grows from 4 to 64. With
\(n=128\), increasing \(k\) from 1 to 128 roughly doubles receive time and
message data. Increasing \(\ell\) from 2 to 32 raises message size from 1.1 to
17.0 KB.

\paragraph{Workflow generation.}
We asked Codex CLI 0.153.4 with \texttt{gpt-6-astra}, at low reasoning effort,
to generate two workflows from short prose specifications, with three runs per
task. Each run started from a fresh project created with
\texttt{zippergen init}. The agent then read the standard skill, ZipperGen's
guide for coding agents. The guide shows a
remote-field example but no temporal operator. We gave no further explanation
of CPL and no reference guard. In the first task, \(L_3\) reported a coin outcome stored at
\(L_0\), while the application payloads contained only an opaque token along
\(L_0\to L_1\to L_2\to L_3\). All three runs generated
\texttt{(At[L0].outcome == "heads") @ L3}, passed validation, and returned the
expected result for both outcomes.

In the second task, \(L_0\) recorded three status updates before the same chain.
A pass remained valid through \(\mathsf{checking}\), a failure cancelled it,
and a later pass restored it. The agent had to recognize this historical
condition and place \texttt{since} under the remote operator. In runs 1 and 2,
it listed the names exported by \texttt{zippergen} and read the \texttt{since}
docstring using \texttt{inspect.getdoc}. The third run first wrote
\texttt{Since}. Validation rejected that name and suggested \texttt{since}.
The same run then placed \(\Since\) outside \(\At{L_0}{(\cdot)}\), thereby
evaluating it along \(L_3\)'s history. The agent's generated tests exposed this
second error, and the agent corrected it. All three final workflows used
\(\At{L_0}{(\mathit{status}\neq\mathsf{failed}\;S\;
\mathit{status}=\mathsf{passed})}\), passed validation, and returned the
expected result for all 27 three-update histories. Together, the generated
workflows illustrate remote-field access and the evaluation of a local
\(\Since\)-formula at a remote lifeline. Episode identity is not tested.

\section{Related Work}
\label{sec:related}

Runtime verification is usually defined over sequential
traces~\cite{LeuckerS09,BauerLS11,BartocciFFR18}. Java-MOP, AspectJ-based
temporal assertions, and Control-Flow Temporal Logic connect specifications
more closely with source code and control flow~\cite{ChenR05,StolzB06,DawesR19}.
Runtime enforcement may prevent violations~\cite{Schneider00}.

Recent runtime checks for LLM-agent systems inspect action or input-output
histories~\cite{WangPS2025,PaduraruBS26,ZhangSELD25,KamathZXUSM25,
AlamdariKM26}. CPL guards instead depend on the causal history visible to the
agent making the decision.

PT-DTL~\cite{SenVAR04} is the closest predecessor. Its remote primitive, inspired
in part by the message-passing logic of Meenakshi and
Ramanujam~\cite{MeenakshiR00}, lets a process refer to formulas and expressions
at the latest causally known state of another process. Its
knowledge-vector monitor carries the required values with application messages.
Scheffel and Schmitz later extend this approach with three-valued monitoring
and give a correctness proof for a formula-only variant of PT-DTL~\cite{ScheffelS14}.
CPL adopts this approach for valued MSCs and
follows Sen et al. in carrying both formula and selected variable values.
Scheffel and Schmitz use the first state for an unseen process, while CPL leaves its
view absent. Lemma~\ref{lem:monitor-invariants} and its Lean proof cover this
setting, with clock entries represented by local event indices.

Singh also places temporal guards on events in multiagent
workflows~\cite{Singh03}, while DistAlgo incrementally
maintains conditions over sent and received messages~\cite{LiuSL17}. CPL also
uses temporal information to guide control flow, but its guards follow PT-DTL's
latest-visible formula and value semantics.

ZipperGen~\cite{BolligFN2026} draws on choreographic programming and
multiparty session types~\cite{HondaYC16,CarboneM13,Montesi2023,ShenKK23,
HirschG22}. This setting also includes projected assertions and runtime
monitoring~\cite{BocchiHTY10,ChenBDHY11,BocchiCDHY17,BurloFS21,HeuvelPD23}.
For example, Scribble/Python monitors endpoint conversations and allows them
to be interrupted~\cite{DemangeonHHNY15}, while Pact studies choices in agent
workflows~\cite{GopinathanFNTB26}. In reversible and adaptive choreographies,
runtime conditions serve a different purpose: they trigger rollback or
replacement~\cite{FrancalanzaMT20,CoppoDV15,DallaPredaGGLM17}.

We use MSCs as our causal model~\cite{itu-msc,AlurY99}. Logics over MSCs and
related traces include~\cite{AdsulGKW24,Mazurkiewicz86,BolligKM2010,BolligFG21}.
Past-CTL gives the
knowable causal past an event-structure semantics and translates it
to distributed field-calculus monitors~\cite{AudritoDSTV22}. Its
modalities quantify over causal paths. A survey of distributed runtime
verification~\cite{FrancalanzaPS18} covers decentralized LTL and
knowledge-based monitoring~\cite{BauerF16,MostafaB15,GrafPQ11}. Vector clocks
also appear in message-sequence enforcement and causal monitoring of
multithreaded programs~\cite{SamadiGK23,LeonHL14,SenVAR06}.

\section{Conclusion}
\label{sec:conclusion}

We integrated PT-DTL-style past-time operators over visible causal history
into workflow guards for choices owned by one lifeline. We proved the monitor correct
for the MSC semantics and mechanized the proof in Lean. Future work will study
more workflows and guard generation from natural-language specifications.

\paragraph{Use of AI assistance.}
AI coding agents (OpenAI Codex and Anthropic Claude) assisted with the Lean
formalization, experiments, and manuscript editing. The author developed and
takes full responsibility for the work. Lean kernel-checks the formalization
without \texttt{sorry} or custom axioms. The experiment artifacts are public.

\paragraph{Acknowledgments.}
This work was partly supported by ``France 2030'', managed by ANR
(ANR-23-PEIA-0006). We thank the anonymous reviewers for comments that
substantially improved the paper.

\bibliographystyle{splncs04-no-url}
\bibliography{lit}

\begin{thebibliography}{10}
\providecommand{\url}[1]{\texttt{#1}}
\providecommand{\urlprefix}{URL }
\providecommand{\doi}[1]{https://doi.org/#1}

\bibitem{AdsulGKW24}
Adsul, B., Gastin, P., Kulkarni, S., Weil, P.: An expressively complete local
  past propositional dynamic logic over {Mazurkiewicz} traces and its
  applications. In: Sobocinski, P., Dal~Lago, U., Esparza, J. (eds.)
  Proceedings of the 39th Annual {ACM/IEEE} Symposium on Logic in Computer
  Science, {LICS} 2024, Tallinn, Estonia, July 8-11, 2024. pp. 2:1--2:13. {ACM}
  (2024)

\bibitem{AlamdariKM26}
Alamdari, P.A., Klassen, T.Q., McIlraith, S.A.: Formal methods meet {LLM}s:
  Auditing, monitoring, and intervention for compliance of advanced {AI}
  systems. In: Proceedings of the 2026 {ACM} Conference on Fairness,
  Accountability, and Transparency ({FAccT} 2026). pp. 4162--4198. ACM,
  Montreal, QC, Canada (Jun 2026)

\bibitem{AlurY99}
Alur, R., Yannakakis, M.: Model checking of message sequence charts. In:
  Baeten, J.C.M., Mauw, S. (eds.) {CONCUR} '99: Concurrency Theory, 10th
  International Conference, Eindhoven, The Netherlands, August 24-27, 1999,
  Proceedings. pp. 114--129. Lecture Notes in Computer Science, Springer (1999)

\bibitem{AudritoDSTV22}
Audrito, G., Damiani, F., Stolz, V., Torta, G., Viroli, M.: Distributed runtime
  verification by past-{CTL} and the field calculus. J. Syst. Softw.
  \textbf{187},  111251 (2022)

\bibitem{BartocciFFR18}
Bartocci, E., Falcone, Y., Francalanza, A., Reger, G.: Introduction to runtime
  verification. In: Bartocci, E., Falcone, Y. (eds.) Lectures on Runtime
  Verification - Introductory and Advanced Topics, pp. 1--33. Lecture Notes in
  Computer Science, Springer (2018)

\bibitem{BauerF16}
Bauer, A., Falcone, Y.: Decentralised {LTL} monitoring. Formal Methods Syst.
  Des.  \textbf{48}(1-2),  46--93 (2016)

\bibitem{BauerLS11}
Bauer, A., Leucker, M., Schallhart, C.: Runtime verification for {LTL} and
  {TLTL}. {ACM} Trans. Softw. Eng. Methodol.  \textbf{20}(4),  14:1--14:64
  (2011)

\bibitem{BocchiCDHY17}
Bocchi, L., Chen, T., Demangeon, R., Honda, K., Yoshida, N.: Monitoring
  networks through multiparty session types. Theor. Comput. Sci.  \textbf{669},
   33--58 (2017)

\bibitem{BocchiHTY10}
Bocchi, L., Honda, K., Tuosto, E., Yoshida, N.: A theory of design-by-contract
  for distributed multiparty interactions. In: Gastin, P., Laroussinie, F.
  (eds.) {CONCUR} 2010 - Concurrency Theory, 21th International Conference,
  {CONCUR} 2010, Paris, France, August 31-September 3, 2010. Proceedings. pp.
  162--176. Lecture Notes in Computer Science, Springer (2010)

\bibitem{BFG2018}
Bollig, B., Fortin, M., Gastin, P.: Gossiping in message-passing systems. CoRR
  \textbf{abs/1802.08641} (2018)

\bibitem{BolligFG21}
Bollig, B., Fortin, M., Gastin, P.: Communicating finite-state machines,
  first-order logic, and star-free propositional dynamic logic. J. Comput.
  Syst. Sci.  \textbf{115},  22--53 (2021)

\bibitem{BolligFN2026}
Bollig, B., Függer, M., Nowak, T.: Provable coordination for {LLM} agents via
  message sequence charts. CoRR  \textbf{abs/2604.17612} (2026), to appear in
  {ISoLA} 2026

\bibitem{BolligKM2010}
Bollig, B., Kuske, D., Meinecke, I.: Propositional dynamic logic for
  message-passing systems. Log. Methods Comput. Sci.  \textbf{6}(3) (2010)

\bibitem{BurloFS21}
Burl{\`{o}}, C.B., Francalanza, A., Scalas, A.: On the monitorability of
  session types, in theory and practice. In: 35th European Conference on
  Object-Oriented Programming ({ECOOP} 2021). Leibniz International Proceedings
  in Informatics, vol.~194, pp. 20:1--20:30. Schloss Dagstuhl --
  Leibniz-Zentrum f{\"{u}}r Informatik (2021)

\bibitem{CarboneM13}
Carbone, M., Montesi, F.: Deadlock-freedom-by-design: multiparty asynchronous
  global programming. In: Giacobazzi, R., Cousot, R. (eds.) The 40th Annual
  {ACM} {SIGPLAN-SIGACT} Symposium on Principles of Programming Languages,
  {POPL} '13, Rome, Italy - January 23 - 25, 2013. pp. 263--274. {ACM} (2013)

\bibitem{ChenR05}
Chen, F., Rosu, G.: {Java-MOP}: {A} monitoring oriented programming environment
  for {Java}. In: Halbwachs, N., Zuck, L.D. (eds.) Tools and Algorithms for the
  Construction and Analysis of Systems, 11th International Conference, {TACAS}
  2005, Held as Part of the Joint European Conferences on Theory and Practice
  of Software, {ETAPS} 2005, Edinburgh, UK, April 4-8, 2005, Proceedings. pp.
  546--550. Lecture Notes in Computer Science, Springer (2005)

\bibitem{ChenBDHY11}
Chen, T., Bocchi, L., Deni{\'{e}}lou, P., Honda, K., Yoshida, N.: Asynchronous
  distributed monitoring for multiparty session enforcement. In: Bruni, R.,
  Sassone, V. (eds.) Trustworthy Global Computing - 6th International
  Symposium, {TGC} 2011, Aachen, Germany, June 9-10, 2011. Revised Selected
  Papers. pp. 25--45. Lecture Notes in Computer Science, Springer (2011)

\bibitem{CoppoDV15}
Coppo, M., Dezani-Ciancaglini, M., Venneri, B.: Self-adaptive multiparty
  sessions. Service Oriented Computing and Applications  \textbf{9}(3--4),
  249--268 (2015)

\bibitem{DallaPredaGGLM17}
Dalla~Preda, M., Gabbrielli, M., Giallorenzo, S., Lanese, I., Mauro, J.:
  Dynamic choreographies: Theory and implementation. Logical Methods in
  Computer Science  \textbf{13}(2) (2017)

\bibitem{DawesR19}
Dawes, J.H., Reger, G.: Specification of temporal properties of functions for
  runtime verification. In: Proceedings of the 34th {ACM/SIGAPP} Symposium on
  Applied Computing ({SAC} 2019). pp. 2206--2214. ACM (2019)

\bibitem{DemangeonHHNY15}
Demangeon, R., Honda, K., Hu, R., Neykova, R., Yoshida, N.: Practical
  interruptible conversations: distributed dynamic verification with multiparty
  session types and {Python}. Formal Methods Syst. Des.  \textbf{46}(3),
  197--225 (2015)

\bibitem{fidge1988timestamps}
Fidge, C.J.: Timestamps in message-passing systems that preserve the partial
  ordering. In: Proceedings of the 11th Australian Computer Science Conference.
  Australian Computer Science Communications, vol.~10, pp. 56--66 (Feb 1988)

\bibitem{FrancalanzaMT20}
Francalanza, A., Mezzina, C.A., Tuosto, E.: Towards choreographic-based
  monitoring. In: Ulidowski, I., Lanese, I., Schultz, U.P., Ferreira, C. (eds.)
  Reversible Computation: Extending Horizons of Computing, Lecture Notes in
  Computer Science, vol. 12070, pp. 128--150. Springer (2020)

\bibitem{FrancalanzaPS18}
Francalanza, A., P{\'{e}}rez, J.A., S{\'{a}}nchez, C.: Runtime verification for
  decentralised and distributed systems. In: Bartocci, E., Falcone, Y. (eds.)
  Lectures on Runtime Verification - Introductory and Advanced Topics, Lecture
  Notes in Computer Science, vol. 10457, pp. 176--210. Springer (2018)

\bibitem{GopinathanFNTB26}
Gopinathan, K., Feser, J., Naim, M., Tavares, Z., Bingham, E.: Pact: A
  choreographic language for agentic ecosystems. In: Proceedings of the 2nd
  International Workshop on Choreographic Programming (CP 2026) (2026), to
  appear. Also available as arXiv:2605.03143

\bibitem{GrafPQ11}
Graf, S., Peled, D.A., Quinton, S.: Monitoring distributed systems using
  knowledge. In: Bruni, R., Dingel, J. (eds.) Formal Techniques for Distributed
  Systems - Joint 13th {IFIP} {WG} 6.1 International Conference, {FMOODS} 2011,
  and 31st {IFIP} {WG} 6.1 International Conference, {FORTE} 2011, Reykjavik,
  Iceland, June 6-9, 2011. Proceedings. pp. 183--197. Lecture Notes in Computer
  Science, Springer (2011)

\bibitem{HavelundR02}
Havelund, K., Ro{\c{s}}u, G.: Synthesizing monitors for safety properties. In:
  Tools and Algorithms for the Construction and Analysis of Systems ({TACAS}
  2002). Lecture Notes in Computer Science, vol.~2280, pp. 342--356. Springer
  (2002)

\bibitem{HeuvelPD23}
van~den Heuvel, B., P{\'{e}}rez, J.A., Dobre, R.A.: Monitoring blackbox
  implementations of multiparty session protocols. In: Katsaros, P., Nenzi, L.
  (eds.) Runtime Verification - 23rd International Conference, {RV} 2023,
  Thessaloniki, Greece, October 3-6, 2023, Proceedings. pp. 66--85. Lecture
  Notes in Computer Science, Springer (2023)

\bibitem{HirschG22}
Hirsch, A.K., Garg, D.: Pirouette: higher-order typed functional
  choreographies. Proc. {ACM} Program. Lang.  \textbf{6}({POPL}),  1--27 (2022)

\bibitem{HondaYC16}
Honda, K., Yoshida, N., Carbone, M.: Multiparty asynchronous session types. J.
  {ACM}  \textbf{63}(1),  9:1--9:67 (2016)

\bibitem{itu-msc}
{ITU-T}: Recommendation {Z.120}: Message sequence chart ({MSC}). Tech. rep.,
  International Telecommunication Union (2011)

\bibitem{KamathZXUSM25}
Kamath, A., Zhang, S., Xu, C., Ugare, S., Singh, G., Misailovic, S.: Enforcing
  temporal constraints for {LLM} agents. CoRR  \textbf{abs/2512.23738} (2025)

\bibitem{Lamport78}
Lamport, L.: Time, clocks, and the ordering of events in a distributed system.
  Commun. {ACM}  \textbf{21}(7),  558--565 (1978)

\bibitem{LeuckerS09}
Leucker, M., Schallhart, C.: A brief account of runtime verification. J. Log.
  Algebraic Methods Program.  \textbf{78}(5),  293--303 (2009)

\bibitem{LiuSL17}
Liu, Y.A., Stoller, S.D., Lin, B.: From clarity to efficiency for distributed
  algorithms. {ACM} Transactions on Programming Languages and Systems
  \textbf{39}(3),  12:1--12:41 (2017)

\bibitem{mattern1989virtual}
Mattern, F.: Virtual time and global states of distributed systems. In:
  Parallel and Distributed Algorithms: Proceedings of the International
  Workshop on Parallel and Distributed Algorithms. pp. 215--226. North-Holland,
  Amsterdam (1989)

\bibitem{Mazurkiewicz86}
Mazurkiewicz, A.W.: Trace theory. In: Brauer, W., Reisig, W., Rozenberg, G.
  (eds.) Petri Nets: Central Models and Their Properties, Advances in Petri
  Nets 1986, Part II, Proceedings of an Advanced Course, Bad Honnef, Germany,
  8-19 September 1986. pp. 279--324. Lecture Notes in Computer Science,
  Springer (1986)

\bibitem{MeenakshiR00}
Meenakshi, B., Ramanujam, R.: Reasoning about message passing in finite state
  environments. In: Montanari, U., Rolim, J.D.P., Welzl, E. (eds.) Automata,
  Languages and Programming, 27th International Colloquium, {ICALP} 2000,
  Geneva, Switzerland, July 9-15, 2000, Proceedings. pp. 487--498. Lecture
  Notes in Computer Science, Springer (2000)

\bibitem{Montesi2023}
Montesi, F.: Introduction to Choreographies. Cambridge University Press (2023)

\bibitem{MostafaB15}
Mostafa, M., Bonakdarpour, B.: Decentralized runtime verification of {LTL}
  specifications in distributed systems. In: 2015 {IEEE} International Parallel
  and Distributed Processing Symposium, {IPDPS} 2015, Hyderabad, India, May
  25-29, 2015. pp. 494--503. {IEEE} Computer Society (2015)

\bibitem{Moura2021}
de~Moura, L., Ullrich, S.: The {Lean 4} theorem prover and programming
  language. In: Platzer, A., Sutcliffe, G. (eds.) Automated Deduction - {CADE}
  28 - 28th International Conference on Automated Deduction, Virtual Event,
  July 12-15, 2021, Proceedings. pp. 625--635. Lecture Notes in Computer
  Science, Springer (2021)

\bibitem{PaduraruBS26}
Paduraru, C., Bouruc, P.L., Stefanescu, A.: A trace-based assurance framework
  for agentic {AI} orchestration: Contracts, testing, and governance. In:
  Proceedings of the 21st International Conference on Evaluation of Novel
  Approaches to Software Engineering ({ENASE} 2026). pp. 420--427. SciTePress
  (2026)

\bibitem{LeonHL14}
{Ponce-de-Le{\'{o}}n}, H., Haar, S., Longuet, D.: Distributed testing of
  concurrent systems: Vector clocks to the rescue. In: Ciobanu, G., M{\'{e}}ry,
  D. (eds.) Theoretical Aspects of Computing - {ICTAC} 2014 - 11th
  International Colloquium, Bucharest, Romania, September 17-19, 2014.
  Proceedings. pp. 369--387. Lecture Notes in Computer Science, Springer (2014)

\bibitem{SamadiGK23}
Samadi, M., Ghassemi, F., Khosravi, R.: Decentralized runtime verification of
  message sequences in message-based systems. Acta Informatica  \textbf{60}(2),
   145--178 (2023)

\bibitem{ScheffelS14}
Scheffel, T., Schmitz, M.: Three-valued asynchronous distributed runtime
  verification. In: 12th {ACM/IEEE} International Conference on Formal Methods
  and Models for Codesign ({MEMOCODE} 2014). pp. 52--61. IEEE (2014)

\bibitem{Schneider00}
Schneider, F.B.: Enforceable security policies. {ACM} Trans. Inf. Syst. Secur.
  \textbf{3}(1),  30--50 (2000)

\bibitem{SenVAR04}
Sen, K., Vardhan, A., Agha, G., Ro{\c{s}}u, G.: Efficient decentralized
  monitoring of safety in distributed systems. In: Proceedings of the 26th
  International Conference on Software Engineering ({ICSE} 2004). pp. 418--427.
  IEEE Computer Society (2004)

\bibitem{SenVAR06}
Sen, K., Vardhan, A., Agha, G., Rosu, G.: Decentralized runtime analysis of
  multithreaded applications. In: 20th International Parallel and Distributed
  Processing Symposium {(IPDPS} 2006), Proceedings, 25-29 April 2006, Rhodes
  Island, Greece. {IEEE} (2006)

\bibitem{ShenKK23}
Shen, G., Kashiwa, S., Kuper, L.: {HasChor}: Functional choreographic
  programming for all (functional pearl). Proc. {ACM} Program. Lang.
  \textbf{7}({ICFP}),  541--565 (2023)

\bibitem{Singh03}
Singh, M.P.: Distributed enactment of multiagent workflows: Temporal logic for
  web service composition. In: The Second International Joint Conference on
  Autonomous Agents and Multiagent Systems, {AAMAS} 2003. pp. 907--914. {ACM}
  (2003)

\bibitem{StolzB06}
Stolz, V., Bodden, E.: Temporal assertions using {AspectJ}. Electronic Notes in
  Theoretical Computer Science  \textbf{144}(4),  109--124 (2006)

\bibitem{WangPS2025}
Wang, H., Poskitt, C.M., Sun, J.: {AgentSpec}: Customizable runtime enforcement
  for safe and reliable {LLM} agents. In: Proceedings of the 48th IEEE/ACM
  International Conference on Software Engineering ({ICSE} 2026). pp. 1--12.
  ACM (2026)

\bibitem{ZhangSELD25}
Zhang, Y., Emma, S.Y., En, A.L.J., Dong, J.S.: {RvLLM}: {LLM} runtime
  verification with domain knowledge. In: Advances in Neural Information
  Processing Systems 38: Annual Conference on Neural Information Processing
  Systems 2025, {NeurIPS} 2025 (2025)

\end{thebibliography}

\end{document}